\documentclass{llncs}
\PassOptionsToPackage{usenames,dvipsnames}{xcolor}
\usepackage{wrapfig}

\usepackage{framed}
\usepackage{amssymb}
\usepackage{amsfonts}
\usepackage{amsmath} 

\usepackage[colorlinks=false]{hyperref}
\usepackage{paralist}

\usepackage{array}
\usepackage{booktabs}
\usepackage{todonotes}
\usepackage[caption=false]{subfig}
\usepackage{wrapfig}

\usepackage[sort&compress]{natbib}
\usepackage{xcolor}
\bibpunct{[}{]}{,}{n}{}{,}   

\makeatletter
\renewcommand\paragraph{\@startsection{paragraph}{4}{\z@}%
                        {-12\p@ \@plus -4\p@ \@minus -4\p@}%
                        {-0.5em \@plus -0.22em \@minus -0.1em}%
                        {\normalfont\normalsize\bfseries}}
\makeatother

\let\llncssubparagraph\subparagraph
\let\subparagraph\paragraph
\let\subparagraph\llncssubparagraph

\newcommand{\toolname}{{\sc Discriminer}\xspace}

\newcommand{\set}[1]{\left\{ #1 \right\}}

\newcommand{\seq}[1]{\langle #1 \rangle}

\newcommand{\TDLP}{TDLP}

\newcommand{\dis}{d}

\newcommand{\LSBZ}{\texttt{LSB0}}
\newcommand{\MSBZ}{\texttt{MSB0}}
\newcommand{\Pat}{\texttt{Pat}}

\def\rmdef{\stackrel{\mbox{\rm {\tiny def}}}{=}} 
\newcommand{\DIST}{{\cal D}}

\usepackage{tikz}
\tikzset{
  treenode/.style = {shape=rectangle, rounded corners,
                     draw, align=center,
                     top color=white, bottom color=blue!20},
  root/.style     = {treenode, font=\Large, bottom color=red!30},
  env/.style      = {treenode, font=\ttfamily\normalsize},
  dummy/.style    = {circle,draw, font=\ttfamily\normalsize}
}
\usepackage[linesnumbered,ruled]{algorithm2e}
\newcommand\tupleof[1]{\left\langle #1 \right\rangle}
\usepackage{listings}
\usepackage{color}

\definecolor{pblue}{rgb}{0.13,0.13,1}
\definecolor{pgreen}{rgb}{0,0.5,0}
\definecolor{pred}{rgb}{0.9,0,0}
\definecolor{pgrey}{rgb}{0.46,0.45,0.48}

\newcommand\lab{\mbox{\textsc{Label}}}

\begin{document}

\title{Discriminating Traces with Time\thanks{This research was
    supported by DARPA under agreement FA8750-15-2-0096.}}
\author{
  Saeid Tizpaz-Niari
  \and
  Pavol {\v C}ern\'y
  \and
  Bor-Yuh Evan Chang\and \\
  Sriram Sankaranarayanan
  \and
  Ashutosh Trivedi
}
\institute{University of Colorado Boulder, USA \\
  \email{\{saeid.tizpazniari,
    pavol.cerny, evan.chang, \\srirams, ashutosh.trivedi\}@colorado.edu}}

\maketitle

\begin{abstract}
  What properties about the internals of a program explain the
  possible differences in 
  its overall running time for different inputs? In this paper, we
  propose a formal framework for 
  considering this question we dub \emph{trace-set discrimination}. 
  We show that even though the algorithmic problem of computing
  maximum likelihood 
  discriminants is NP-hard, approaches based on
  integer linear programming (ILP) and decision tree learning can be useful in
  zeroing-in on the program internals. On a set of Java benchmarks,
  we find that compactly-represented decision trees scalably discriminate with
  high accuracy---more scalably than maximum likelihood discriminants and with
  comparable accuracy. We demonstrate on three larger
  case studies how decision-tree discriminants produced by our tool are
  useful for debugging timing 
  side-channel vulnerabilities (i.e., where a malicious observer
  infers secrets simply from passively watching execution times) and
  availability vulnerabilities.
\end{abstract}

\section{Introduction}
\label{sec:introduction}


Different control-flow paths in a program can have varying execution times. 
Such observable differences in execution times may be explainable by
information about the 
program internals, such as whether or not a given function or
functions were called. 
How can a software developer (or security analyst)
determine what internals may or may not explain the varying
execution times of the program? In this paper, we consider the problem
of helping developers and analysts to identify such explanations. 

We identify a core problem for this task---the {\em trace-set discrimination} problem.
Given a set of execution traces with observable execution times binned (or
clustered) into a finite set of labels, a \emph{discriminant} (or classifier) is
a map relating each label to a property (i.e., a Boolean formula) satisfied by
the traces assigned to that label. Such a discriminant model can then be used,
for example, to predict a property satisfied by some trace given the timing
label of that trace.


This problem is, while related, different than the profiling problem. In
performance profiling, the question is given an execution trace, how do the
various parts of the program contribute to the overall execution time?
The trace-set discrimination problem, in 
contrast, looks for distinguishing features among multiple traces that result in
varying execution times.

Crucially, once we can explain the timing
differences in terms of properties of traces (e.g., what functions are
called only in traces with long execution time), the analyst can use the explanation to diagnose the possible
timing side-channel and potentially find a fix for the vulnerability. 
Section~\ref{sec:overview} shows on an example how a security analyst might
use the tool for debugging information leaks. 


In this paper, we consider the discriminating properties of traces to be Boolean
combinations of a given set of atomic predicates. These atomic predicates
correspond to actions that can be observed through instrumentation in a training
set of execution traces. Examples of such predicates are as follows:
\begin{inparaenum}[(1)]%
\item Does the trace have a call to the function $f$ in the program?
\item Does the trace have a call to the \texttt{sort} function with an array of
more than a $1000$ numbers?
\end{inparaenum}
In our case study, we consider atomic predicates corresponding to the number
of times each function is called.

Concretely, our overall approach is to first obtain a set of execution traces
with information recorded to determine the satisfiability of the given atomic
predicates along with corresponding execution times. Then, we cluster these
training traces based on their overall execution times to bin them into timing
labels. Finally, we learn a trace-set discriminant model from these traces
(using various techniques) to capture what is common amongst the traces with the
same timing labels and what is different between traces with different labels.


In particular, we make the following contributions:
\begin{itemize}\itemsep 0pt
\item We formalize the problem of \emph{trace-set discrimination} with timing
differences and show that the algorithmic problem of finding the maximum
likelihood conjunctive discriminant is NP-hard (Section~\ref{sec:problem}).
\item We describe two methods for learning trace-set discriminants:%
  \begin{inparaenum}[(1)]
    \item a direct method for inferring the maximum likelihood conjunctive
    discriminant using an encoding into integer linear programming (ILP) and
    \item by applying decision tree learning
  \end{inparaenum}
that each offer different trade-offs (Section~\ref{sec:mining}). For instance,
decision tree algorithms are designed to tolerate noisy labels and work
effectively on large data sets but do not have formal guarantees. On a set of
microbenchmarks, we find that the methods have similar accuracy but decision
tree learning appears more scalable.
\item We present three case studies in identifying and debugging timing
side-channel and availability vulnerabilities, armed with a prototype tool \toolname{} that
performs label clustering and decision tree-discriminant learning
(Section~\ref{sec:experimental}). These case studies were conducted on
medium-sized Java applications, which range in size from approximately 300 to
3,000 methods and were developed by a third party vendor as challenge problems
for identifying and debugging such side-channel vulnerabilities. We show that
the decision trees produced by \toolname{} are useful for explaining the timing
differences amongst trace sets and performing this debugging task.
\end{itemize}


In our approach, we need to execute both an
instrumented and an uninstrumented version of the program of interest on
the same inputs. This is because a trace of
the instrumented program is needed to determine the satisfiability of the atomic
predicates, while the execution time of interest is for the uninstrumented
program. Therefore we need to assume that the program is deterministic. 
Since timing observations are noisy due to many sources of
non-determinism, each trace is associated with a \emph{distribution} over the
labels. For instance, a trace may have a label $\ell_1$ with probability $0.9$
and label $\ell_2$ with probability $0.1$.

Like with profiling, we also assume the test inputs that drive the program of
interest to expose interesting behavior are given. It is a separate problem to
get such interesting inputs: whether the analyst has logged some suspicious
inputs from a deployment or whether the developer generates tests using random
or directed test-case generation.

\section{Timing Side-Channel Debugging with \toolname}
\label{sec:overview}
In this section, we demonstrate by example how \toolname{} can be useful in
identifying timing side-channel vulnerabilities
and suggesting ways to fix them.
We use an application called
SnapBuddy\footnote{From DARPA STAC
({\scriptsize\url{www.darpa.mil/program/space-time-analysis-for-cybersecurity}}).}
as an example.
SnapBuddy
is a Java application with 3,071 methods, implementing a mock social network in which
each user has their own page with a photograph.

\subsubsection*{Identifying a Timing Side-Channel with Clustering.}

The analyst interacts with the application by issuing download
requests to the pages of various users to record execution times.
Figure~\ref{fig:sbtime} shows a scatter plot of the running times of
various traces with each trace represented by a point in the figure.  The
running times are clustered into $6$ different clusters using a
standard $k$-means clustering algorithm and shown using different
colors. We see that for some users, the download times were roughly $15$ seconds, whereas for
some others they were roughly $7.5$ seconds. This significant time differential suggests a
potential timing side-channel if the difference can be correlated with sensitive
program state and thus this differential should be investigated further with
\toolname{}.
To see how such a time differential could be a timing side-channel,
let us consider an attacker that
(a) downloads the public profile pages of all users and learns each download time,
and (b) can observe timing between packets by sniffing the
network traffic between legitimate users.
If the attacker observes user Alice downloading the page of another user whose identity
is supposed to be a secret
and sees that the download took approximately $7.5$ seconds, the attacker can infer
that Alice downloaded the page of one of the six users corresponding
to the six squares (with time close to 7.5 seconds) in Figure~\ref{fig:sbtime}. The timing
information leak thus helped the attacker narrow down the
possibilities from hundreds of users to six.

\subsubsection*{Debugging Timing Side-Channels with Decision Tree Learning.}

How can the analyst go about debugging the SnapBuddy application to eliminate
this timing side-channel?
We show how \toolname can help. Recall that the analyst downloaded pages
of all the users. Now the same download queries are executed over an
instrumented version of the SnapBuddy server to record the number of
times each method in the application is called by the trace. As a
result, we obtain a
set of traces with their (uninstrumented) overall running times and set
of corresponding method calls.

\begin{figure}[t]
\centering
\begin{minipage}[b]{0.47\textwidth}
  \centering
  \includegraphics[width=\textwidth]{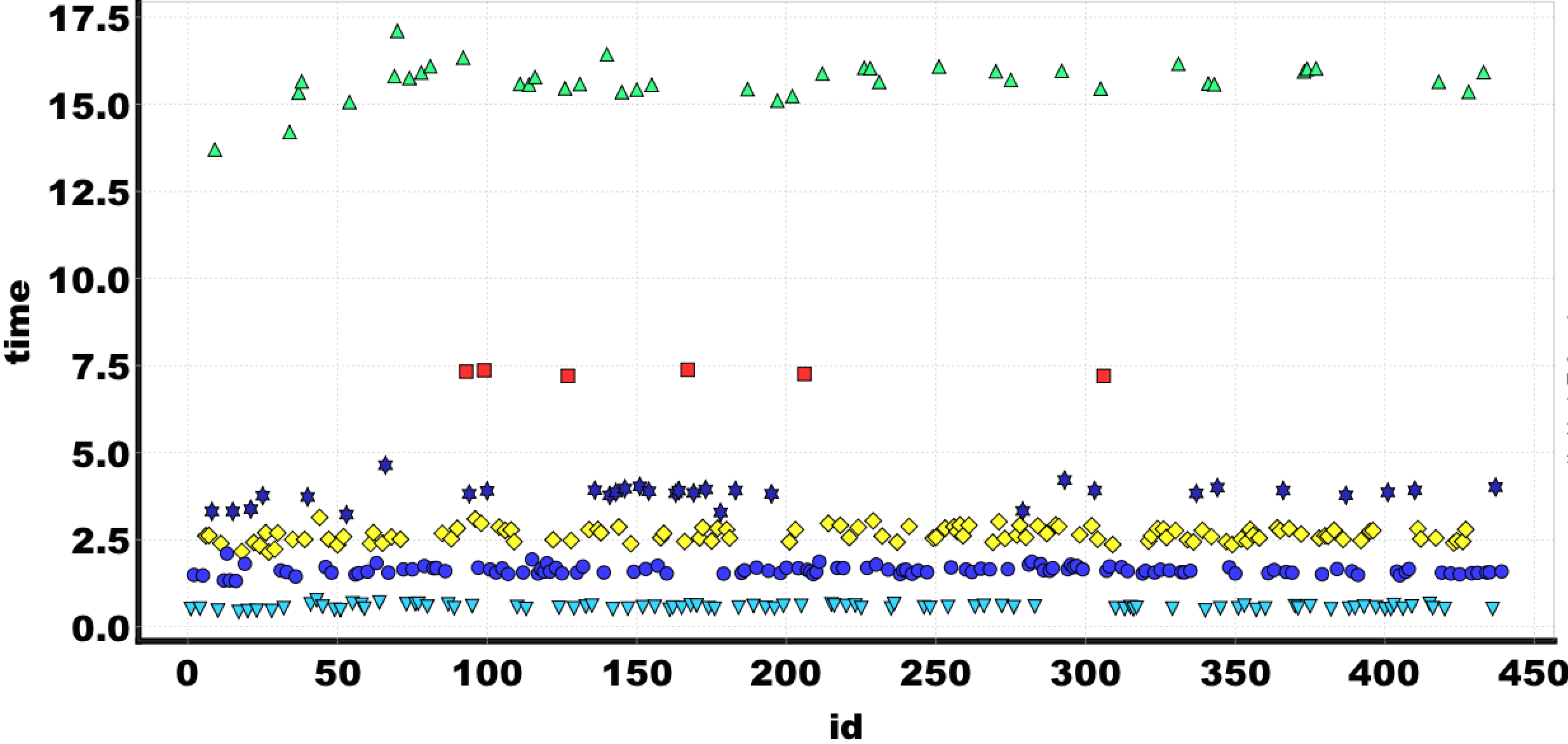}
  \caption{Cluster running times from the SnapBuddy to produce labels. The scatter plot shows a differential corresponding to a possible timing side-channel.}
  \label{fig:sbtime}
\end{minipage}\hfill
\begin{minipage}[b]{0.50\textwidth}
  \centering
  \includegraphics[width=\textwidth]{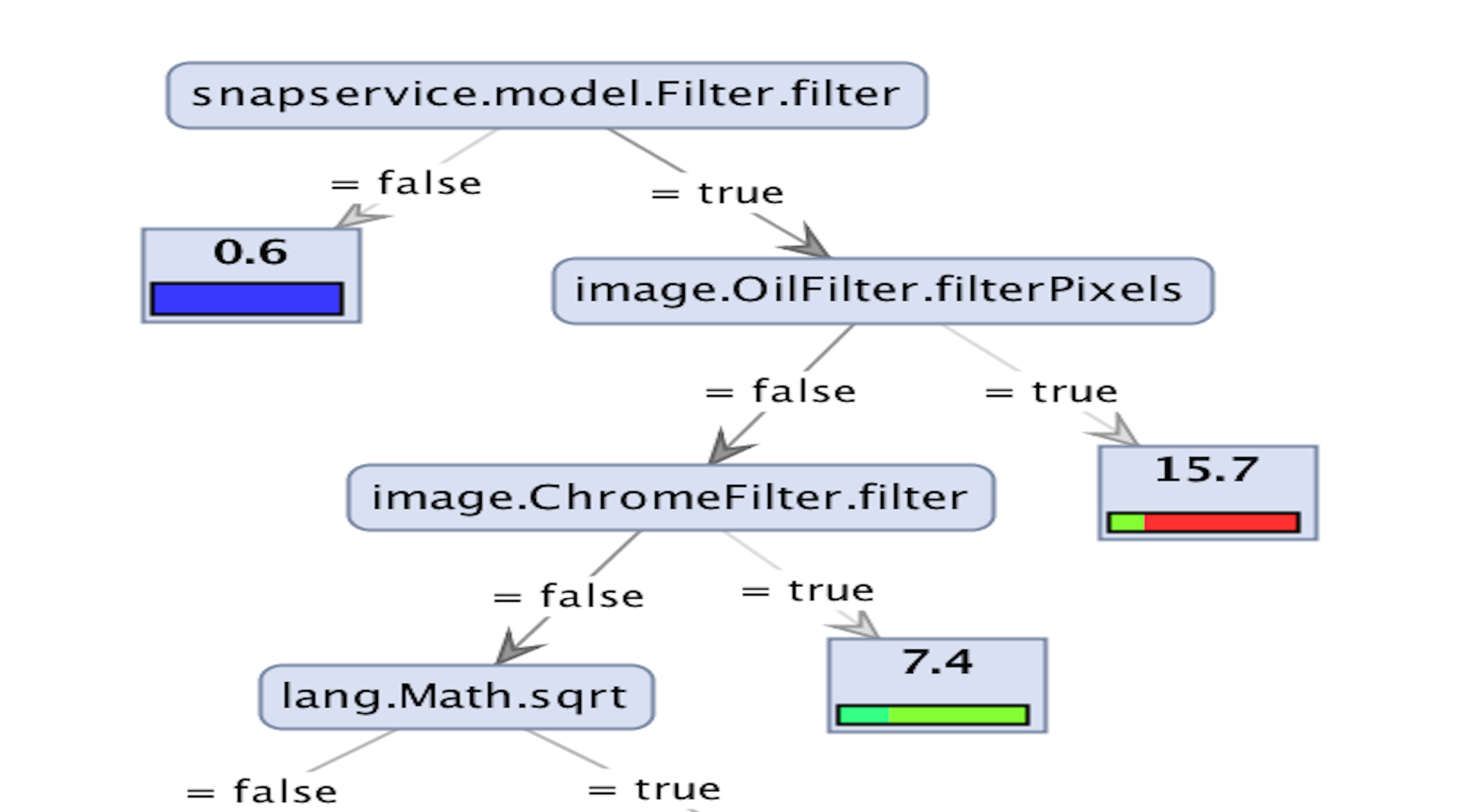}
  \caption{Snippet of a decision-tree discriminant learned from SnapBuddy traces using the timing labels from Figure~\ref{fig:sbtime}.}
  \label{fig:sbdectree}
\end{minipage}
\vspace{-2em}
\end{figure}
Then \toolname uses the
standard \emph{CART} decision tree learning
algorithm~\cite{BFOS84} to infer a decision tree that
succinctly represents
a discriminant
using atomic predicates that characterize
whether or not the trace invoked a particular method (shown in Figure~\ref{fig:sbdectree}).
For instance, the cluster representing the longest running time (around $15$ seconds) is
discriminated by the property ${\small\texttt{snapservice.model.Filter.filter}}
\wedge {\small\texttt{image.OilFilter.filterPixels}}$, indicating that the two methods
are both invoked by the trace. Likewise, the cluster representing the
running time around $7.5$ seconds is discriminated by the
property ${\small\texttt{snapservice.model.Filter.filter}}
\wedge \lnot{\small\texttt{image.OilFilter.filterPixels}}
\wedge {\small\texttt{image.ChromeFilter.filter}}$, indicating that
{\small\texttt{image.OilFilter.filterPixels}} must not be invoked
while the other two must be.

The analyst might now suspect what is going on: the timing
differences are caused by the filters that each user chooses to apply
to their picture. 
Note that the analyst running \toolname did not need to
know that the filters are important for causing this time differential, or even that
they existed. The tool discovers them simply because the trace
contains all method calls, and the decision tree learning algorithm
produces a useful discriminant.

A possible fix now suggests itself: make sure that the execution of
each type of filter takes the same amount of time (though of course an
implementation of such a fix still requires development effort). 
Overall, the example demonstrates how the decision tree produced by
\toolname{} can be used to debug (and potentially fix) side-channel
vulnerabilities.

\section{Trace-Set Discrimination Problem}
\label{sec:problem}
 A \emph{discrete probability distribution}, or just distribution,
over a finite set $L$ is a function $\dis : L {\to} [0, 1]$ such that
$\sum_{\ell \in L} \dis(\ell) = 1$.  Let $\DIST(L)$ denote the set of all
discrete distributions over $L$.

Let $p_1,\ldots,p_m$ represent a set of \emph{atomic predicates} over
traces. Each predicate evaluates to a Boolean value over a given
trace.  Therefore, for simplicity, we represent a trace simply by the
truth valuations of the predicates over the trace. In addition to
atomic predicates, traces are associated with a distribution over
labels.
These distributions are generated by first measuring the execution time $t$ of
the trace. The execution
time is obtained as the average over some fixed number of measurements
$M> 0$.  Therefore, the timing is taken to be a Gaussian random
variable with mean $t$ and a standard deviation $\sigma_t$. Using this
information, we derive a discrete distribution $d \in \DIST(L)$ over
the set of labels in $L$. 

\begin{definition}[Traces, Predicates and Label Distributions]
An execution trace $T$ of the program is a tuple
$\tupleof{\tau,d}$ wherein
$\tau = \ \tupleof{\rho_1,\ldots,\rho_m}$ represents the truth
valuations to the predicates $p_1, \ldots, p_m$, respectively and
$d \in \DIST(L)$ is the associated label distribution over the finite
set of labels $L$.
\end{definition}

We define a \emph{trace discriminant} as a tuple of Boolean formulae that
predict the labels of the traces given the truth valuations in the following
fashion. 

\begin{definition} 
Given a set of labels $L = \set{ \ell_1, \ldots, \ell_K}$ and predicates
$P = \set{p_1, \ldots, p_m}$, a \textbf{discriminant} $\Psi$ is a tuple
$\tupleof{\varphi_1, \ldots, \varphi_K}$ of Boolean formulae where each formula
$\varphi_i$ is over the predicates in $P$ and corresponds to a label $\ell_i$.
\end{definition}
A trace $\tupleof{\tau, d}$ receives a label $\ell_k$ under trace discriminant
$\Psi = \tupleof{\varphi_1,\ldots,\varphi_K}$, and we write
$\lab(\tupleof{\tau, d}, \Psi) = \ell_k$, if $k$ is the
smallest index $1 \leq i \leq K$ such that $\tau \models
\varphi_i$, i.e. $\varphi_i$ evaluates to \texttt{true} for the truth valuation
$\tau$.
Formally, 
\[
\lab(\tupleof{\tau, d}, \Psi) =
\begin{cases}
  \ell_1 & \texttt{if } \tau \models \varphi_1, \texttt{ else}\\ 
  \ell_2 & \texttt{if } \tau \models \varphi_2, \texttt{ else}\\
  \vdots & \vdots \\ 
  \ell_K & \texttt{if } \tau \models \varphi_K.
  \end{cases}
\]

\begin{definition}
Given a set of predicates $\set{p_1, \ldots, p_m}$, set of labels
$\{ \ell_1,\ldots,\ell_K\}$, and a set of traces
 $\set{\tupleof{\tau_1,d_1},  \ldots, \tupleof{\tau_N, d_N}}$, the \textbf{trace set
discriminant problem} $(\TDLP{})$ is to learn a trace discriminant $\Psi
=\ \tupleof{\varphi_1,\ldots,\varphi_K}$. 
\end{definition}

In general, there are numerous  possible
discriminants that can be inferred for a given instance of
the \textsc{tdlp}. We consider two approaches in this
paper: (a) a \emph{formal} maximum likelihood
learning model over a structured set of discriminants and
(b) an informal decision tree learning approach to maximize accuracy
while minimizing the discriminant size.

\subsection{Maximum Likelihood Learning}
Given a discriminant and a set of traces, we define the
likelihood of the discriminant as the probability that each trace
$\tupleof{\tau_i,d_i}$ receives the label $\lab(\tupleof{\tau_i,d_i}, \Psi)$
dictated by the discriminant.  
\begin{definition}
The \textbf{likelihood} $\lambda(\Psi)$ of a discriminant $\Psi$ over a set of traces
$\set{\tupleof{\tau_1,d_1},  \ldots, \tupleof{\tau_N, d_N}}$ is given by
$\lambda(\Psi) = \ \mathop{\prod}_{i=1}^N d_i \left( \lab(\tupleof{\tau_i,d_i},
\Psi) \right)\,$.
\end{definition}

\noindent The \emph{maximum likelihood} discriminant $\Psi_{ml}$ is defined as the
discriminant amongst all possible Boolean formulae that maximizes
$\lambda(\Psi)$, i.e. 
$\Psi_{ml} =  \mbox{argmax}_{\Psi}\left( \lambda(\Psi) \right)$.
This maximization runs over the all possible tuples of $K$ Boolean formulae over $m$
atomic predicates, i.e, a space of $(K!)\binom{2^{2^m}}{K}$ possible
discriminants!
In particular, Hyafil and Rivest~\cite{hyafil1976constructing} show that the
problem of learning optimal decision trees is NP-hard.
Therefore, for our formal approach, we consider the following
simpler class of discriminants by restricting the form of the Boolean formulae
$\varphi_j$ that make up the discriminants to monotone conjunctive formulae.
\begin{definition}[Conjunctive Discriminants]\label{Def:conjunctive-discriminant} 
  A monotone \emph{conjunctive} formula over predicates $P = \set{p_1, \ldots, p_m}$
  is a finite conjunction of the form $\bigwedge_{j=1}^r p_{i_j}$ such that
  $1 \leq i_1, \ldots, i_r \leq m$.
  A discriminant $\Psi = \tupleof{\varphi_1,\ldots,\varphi_K}$ is a
  (monotone) conjunctive discriminant if each $\varphi_i$ is a monotone
  conjunctive formula for $1 \leq i \leq K$. 
  In  order to make a traces discriminant exhaustive, we assume $\varphi_K$ to
  be the formula $\texttt{true}$. 
\end{definition}
The number of conjunctive discriminants is $(K-1)! \binom{2^m}{K-1}$.
However, they can be easily represented and learned using SAT or ILP solvers, as shown
subsequently.
Moreover, working with simpler monotone conjunctive discriminants is preferable~\cite{Domingos1999}
in the presence of noisy data, as using formal maximum likelihood model to learn
arbitrary complex Boolean function would lead to over-fitting.
The problem of \emph{maximum likelihood} conjunctive discriminant
is then naturally defined.
We refine the result of~\cite{hyafil1976constructing} in our context to show
that the problem of learning (monotone) conjunctive discriminants is already
NP-hard.

\begin{theorem}\label{Theorem:npc-conj-discriminant}
  Given an instance of \textsc{tdlp}, the problem of finding the maximum likelihood
  conjunctive discriminant is \textsc{NP}-hard. 
 \end{theorem}
\begin{proof}
  We prove the NP-hardness of the problem of finding maximum likelihood
  conjunctive discriminant by giving a reduction from the \emph{minimum weight
    monotone SAT problem} that is already known to be NP-hard.
    Recall that a monotone Boolean formula is propositional logic formula 
  where all the literals are positive.
  Given a monotone instance of SAT $\phi = \bigwedge_{j = 1}^{n} C_j$ over
  the set of variable $X = \set{x_1, \ldots, x_m}$, the minimum weight
  monotone SAT problem is to find a truth assignment satisfying $\phi$ with as
  few variables set to $\texttt{true}$ as possible.    
  
  Consider the trace-set discrimination problem $P_\phi$ where there is one
  predicate $p_i$ per variable $x_i \in X$ of $\phi$, two labels $\ell_1$ and
  $\ell_2$, and the set of traces such that
  \begin{itemize}
  \item
    there is one trace $\seq{\tau_j, d_j}$ per clause $C_j$ of $\phi$  such that
    predicate $p_i$ evaluates to true in the trace $\tau_j$ if variable $x_i$
    \emph{does not occur} in clause $C_j$, and the label distribution $d_j$ is
    such that $d_j(\ell_1) = 0$ and $d_j(\ell_2) = 1$.
  \item
    there is one trace $\seq{\tau^i, d^i}$ per variable $x_i$ of $\phi$
    such that only the predicate $p_i$ evaluates to false in the trace $\tau^i$
    and the label distribution $d^i$ is such that $d^i(\ell_1) = 1-\varepsilon$
    and $d^i(\ell_2) = \varepsilon$ where $0 < \varepsilon < \frac{1}{2}$.
  \end{itemize}
  Observe that for every truth assignment $(x_1^*, \ldots, x_m^*)$ to variables
  in $X$, there is a 
  conjunctive discriminant $\wedge_{x_i^*=1} p_i$ such that if the clause $C_j$ is
  satisfied then the trace $\seq{\tau_j, d_j}$ receives the label $\ell_2$.
  This implies that the likelihood of the discriminant is non-zero only for the
  discriminant corresponding to satisfying valuations of $\phi$. 
  Moreover, for every variable $x_i$ receiving a true assignment, the trace
  $\seq{\tau^i, d^i}$ receives the label $\ell_2$ with $\varepsilon$ contributed
  to the likelihood term and
  for every variable $x_i$ receiving false assignment, the trace
  $\seq{\tau^i, d^i}$ receives the label $\ell_1$ with  $1-\varepsilon$ being
  contributed to the likelihood.
  This construction implies that a maximum likelihood discriminant should give
  label $\ell_2$ to all of the traces $\seq{\tau_j, d_j}$ and label $\ell_1$ to
  as many traces in $\set{\tau^i, d^i}$ as possible.
  It is easy to verify that there exists a truth assignment of size $k$
  for $\phi$ if and only if there exists a conjunctive discriminant
  in $P_\phi$ with likelihood $\prod_{i=1}^{k} \varepsilon \cdot \prod_{i=1}^{m-k} (1-\varepsilon)$.
\qed
\end{proof}

\subsection{Decision Tree Learning}
As noted earlier, the max likelihood approach over structured Boolean formulae can be
prohibitively expensive when the number of traces, predicates and labels are large.
An efficient alternative is to consider decision tree learning approaches that
can efficiently produce accurate discriminants while keeping the size of the
discriminant as small as possible.
The weighted accuracy of a discriminant $\Psi$ over traces
$\tupleof{\tau_i,d_i}, i=1,\ldots, N$ is defined additively as
$ \alpha(\Psi):\ \frac{1}{N} \sum_{i=1}^N d_i\left( \lab( \tupleof{\tau_i,d_i}, \Psi) \right)$.
This accuracy is a fraction between $[0,1]$ with higher accuracy
representing a better discriminant.

A decision tree learning algorithm seeks to learn a discriminant as a
decision tree over the predicates $p_1,\ldots,p_m$ and outcome labels
$\ell_1,\ldots,\ell_K$.  Typically, algorithms will maximize
$\alpha(\Psi)$ while keeping the description length $|\Psi|$ as small
as possible. A variety of efficient tree learning algorithms have
been defined including ID3~\cite{Quinlan/1986/ID3}, CART~\cite{Breiman/1984/CART},
CHAID~\cite{Kass80} and many others~\cite{TAN06,MRT12}. These
algorithms have been supported by popular machine learning
tools such as Scikit-learn python library
(\url{http://scikit-learn.org/stable/}) and RapidMiner~\cite{AKT12}.

\section{Discriminant Analysis}
\label{sec:mining}
In this section, we provide details of max likelihood and decision tree
approaches, and compare their performances over a scalable set of micro-benchmarks. 
\subsection{Maximum Likelihood Approach}
We now present an approach for inferring a conjunctive discriminant
$\Psi$ using integer linear programming (ILP) that maximizes the
likelihood $\lambda(\Psi)$ for given predicates $p_1,\ldots,p_m$,
labels $\ell_1,\ldots, \ell_K$ and traces $\tupleof{\tau_1,d_1}$,
$\ldots$, $\tupleof{\tau_N,d_N}$. This problem was already noted to
be NP-hard in Theorem~\ref{Theorem:npc-conj-discriminant}.

We first present our approach for the special case of $K=2$ labels.
Let $\ell_1, \ell_2$ be the two labels. Our goal is to learn a
conjunctive formula $\varphi_1$ for $\ell_1$.
We use binary decision variables $x_1,\ldots,x_m$ wherein $x_i = 1$ denotes that
$\varphi_1$ has the predicate $p_i$ as a conjunct, whereas $x_i = 0$ denotes that
$p_i$ is not a conjunct in $\varphi_1$. Also we add binary decision
variables $w_1,\ldots,w_N$ corresponding to each of the $N$ traces,
respectively.  The variable $w_i = 1$ denotes that the trace
$\tupleof{\tau_i,d_i}$ receives label $\ell_2$ under $\varphi_1$ and
$w_i= 0$ indicates that the trace receives label $\ell_1$. The likelihood of the
discriminant $\Psi$ can be given  as 
$ \lambda(\Psi) \rmdef \ \prod_{i=1}^N \left\{ \begin{array}{cc}
  d_i(\ell_1) & \mbox{if}\ w_i = 0 \\
  d_i(\ell_2) & \mbox{if}\ w_i = 1 \end{array}\right.\,.$
Rather than maximize $\lambda(\Psi)$, we equivalently maximize
$\log(\lambda(\Psi))$
\[
\log(\lambda(\Psi)) =
\sum_{i=1}^N \left\{ \begin{array}{cc}
  \log(d_i(\ell_1)) & \mbox{if}\ w_i = 0 \\
  \log(d_i(\ell_2)) & \mbox{if}\ w_i = 1
\end{array}\right.\,.
\]
Let $r_i := d_i(\ell_1) = 1 - d_i(\ell_2)$, and simplify the expression for
$\log(\lambda(\Psi))$ as $\sum_{i=1}^N (1-w_i) \log(r_i) + w_i \log(1-r_i)$.

Next, the constraints need to relate the values of $x_i$ to each $w_i$. Specifically,
let for each trace $\tupleof{\tau_i,d_i}$, $R_i \subseteq \{ p_1,\ldots,p_m\}$ denote
the predicates that are valued \emph{false} in the trace. We can verify that if $w_i = 0$, then
none of the predicates in $R_i$ can be part of $\varphi_1$, and if $w_i=1$, at least
one of the predicates in $R_i$ must be part of $\varphi_1$. This is expressed by the
following inequality $\frac{1}{|R_i|} ( \sum_{p_k \in R_i} x_k) \leq\ w_i\ \leq \sum_{p_k \in R_i}
x_k $.
If any of the $p_k \in R_i$ is included in the conjunction, then the LHS of the inequality
is at least $\frac{1}{|R_i|}$, forcing $w_i = 1$. Otherwise, if all $p_k$ are not included,
the RHS of the inequality is $0$, forcing $w_i = 0$.

\noindent The overall ILP is given by
\begin{eqnarray}\label{eq:lp-problem}
  \max & \sum_{i=1}^N (1-w_i) \log(r_i) + w_i \log(1-r_i) \nonumber \\
  \mathsf{s.t.} & \frac{1}{|R_i|} ( \sum_{p_k \in R_i} x_k) \leq\ w_i & i =1,\ldots,N \nonumber\\
  &  w_i\ \leq \sum_{p_k \in R_i} x_k & i = 1,\ldots,N \nonumber \\
  & x_j \in \{0,1\},\ w_i \in \{0,1\} & i = 1,\ldots, N,\ j = 1,\ldots, m
\end{eqnarray}

\begin{theorem}
  Let $x_1^*,\ldots,x_m^*$ denote the solution for ILP~\eqref{eq:lp-problem}
  over a given TDLP instance with labels $\{\ell_1,\ell_2\}$. 
  The discriminant $\Psi = \tupleof{\varphi_1, \texttt{true}}$ wherein
  $\varphi_1 = \ \bigwedge_{x_i^* = 1} p_i$ maximizes the likelihood
  $\lambda(\Psi)$ over all conjunctive discriminants. 
\end{theorem}

With the approach using the ILP in Eq.~\eqref{eq:lp-problem}, we can tackle an instance with $K > 2 $ labels by
recursively applying the two label solution. First, we learn a formula $\varphi_1$ for $\ell_1$
and $L \setminus \ell_1$. Next, we eliminate all traces that satisfy $\varphi_1$ and eliminate the label
$\ell_1$. We then recursively consider $\hat{L}:\ L \setminus \ell_1$ as the new label set. Doing so,
we obtain a discriminant $\Psi:\ \tupleof{\varphi_1, \varphi_2, \ldots,
  \varphi_{K-1}, \texttt{true}}$.

In theory, the ILP in~\eqref{eq:lp-problem} has $N+m$ variables, which
can be prohibitively large. However, for the problem instances
considered, we drastically reduced the problem size through standard
preprocessing/simplification steps that allowed us to resolve the
values of $x_i, w_j$ for many of the variables to constants.

\subsection{Decision Tree Learning Appraoch}
In order to discriminate traces, \toolname employs decision tree learning to
learn classifiers that discriminate the traces.
Given a set of $N$ traces on a dependent variable (labels)  $L$
that takes finitely-many values in the domain $\set{\ell_1, \ldots, \ell_K}$ and $m$
feature variables (predicates) $F = \set{f_1, \ldots, f_m}$, the goal of a
classification algorithm is to produce a partition the space of the feature variables
into $K$ disjoint sets $A_1, \ldots, A_K$ such that the predicted value of $L$
is $i$ if the $F$-variables take value in $A_i$. 
Decision-tree methods yield rectangular sets $A_i$  by recursively partitioning
the data set one $F$ variable at a time.
CART (\emph{Classification and Regression Trees}) is a popular and effective
algorithm to learn decision-tree based classifiers.
It constructs binary decision trees by iteratively exploring features and
thresholds that yield the largest information gain (Gini index) at each node.
For a detailed description of the CART, we refer to~\cite{BFOS84}. 

\subsection{Performance Evaluation}
We created a set of micro-benchmarks---containing a side-channel in
time---to evaluate the performance of the decision-tree discriminator
computed using \textit{scikit-learn} implementation of CART and
the maximum likelihood conjunctive discriminant using an ILP implementation from
the GLPK library.

These micro-benchmarks consist of a set of programs that take as an input a
sequence of binary digits (say a secret information), and perform some computation whose execution time
(enforced using \texttt{sleep} commands) depends on some property of
the secret information.
For the micro-benchmark series $\LSBZ$ and $\MSBZ$,
the execution time is a Gaussian-distributed random variable whose
mean is proportional to the position of least significant $0$ and most significant $0$
in the secret, respectively. In addition, we have a micro-benchmark series
$\Pat_d$ whose execution time is a random variable whose mean
depends upon the position of the pattern $d$ in the input.
For instance, the micro-benchmark $\Pat_{101}$ takes a $20$-bit input data
and the leftmost occurrence $i$ of the pattern $101$ executes three 
methods $F_i, F_{i+1}, F_{i+2}$  with mean exec. time of a method $F_j$
being $10{*}j$ ms. 

In our experiments with micro-benchmarks, we generate the dataset by randomly
generating the input.
For each input, we execute the benchmark programs $10$ times to approximate the
mean and the standard deviation of the observation, and log the list of method called
for each such input. 
For a given set of execution traces, we cluster the execution time based on
their mean and assign weighted labels to each trace according to Gaussian
distribution.
We defer the details of this data collection to Section~\ref{sec:experimental}.  
Our dataset consists of trace id, label, weight, and method calls for every
execution trace.
We use this common dataset to both the decision-tree and the
maximum likelihood algorithms.  
\begin{table}[t]
  \centering
  \caption{Micro-benchmark results for decision-tree discriminators learned
    using decision tree and the max-likelihood approach. 
    Legend: \textbf{\#M}: number of methods, \textbf{\#N:} number of traces,
    \textbf{T}:\ computation time in seconds, \textbf{A}: accuracy, 
    \textbf{H}: decision-tree height, \textbf{M}: max. discriminant size
    (Max. \# of conjuncts in discriminants), \textbf{$\epsilon < 0.1$ sec.}
   }
  \label{table4-1}
  \begin{tabular}{ || l | r | r || r | r | r || r | r | r ||}
    \hline
    &       &    & \multicolumn{3}{c||}{Decision Tree} & \multicolumn{3}{c||}{Max-Likelihood} \\
    \cline{4-9}
    ~~~~Benchmark ID~~~~~~ & \# \textbf{M} & \#\textbf{N} & \textbf{T} & \textbf{A} & \textbf{H} & \textbf{T} & \textbf{A} & $\textbf{M}$ \\ \hline
    \LSBZ & 10 & 188 & $\epsilon$ & 100\% & 7 & $\epsilon$ & 100 \% &  10 \\ \hline
    \MSBZ & 10 & 188 & $\epsilon$ & 100\% & 7 & $\epsilon$ & 100 \% &  10 \\ \hline
    $\Pat_{101}$ & 20 & 200 & $\epsilon$ & 100\% & 13 & 0.2 & 89.4\% &  20 \\ \hline
    $\Pat_{1010}$ & 50 & 500 & $\epsilon$ & 98.4\% & 22 & 1.3 & 93.6\%  & 50 \\ \hline
    $\Pat_{10111}$ & 80 & 800 & 0.1 & 97.8\% & 37 & 8.1 & 94.8\% &  72 \\ \hline
    $\Pat_{10101}$ & 100 & 1000 & 0.2 & 92.9\% & 43 & 9.8 & 87.9\% &  86 \\ \hline
    $\Pat_{10011}$ & 150 & 1500 & 0.5 & 89.2\% & 44 & 45.0 & 91.5\% & 118 \\ \hline
    $\Pat_{101011}$ & 200 & 2000 & 0.8 & 92.1\% & 50 & 60.2 & 90.9\%  & 156 \\ \hline
    $\Pat_{1010101}$ & 400 & 4000 & 4.2 & 88.6\% & 111 & 652.4 & 92.9\% & 294 \\ \hline
    \end{tabular}
\vspace{-2em}
\end{table}

Table~\ref{table4-1} shows the performance of the decision-tree
classifiers and the max-likelihood approach for given micro-benchmarks.
The table consists of benchmark scales (based on the number of
methods and traces), the accuracy of approaches, time of computing decision tree
and max-likelihood discriminant, the
height of decision tree, and the maximum number of 
conjuncts among all learned discriminants in the max-likelihood approach. 
In order to compute the performance of both models and avoid overfitting, we train
and test data sets using group $k$-fold cross-validation procedure with $k$
set to $20$.

Table~\ref{table4-1}
shows that both decision tree and max-likelihood approaches have decent
accuracy in small and  medium sized benchmarks. 
On the other hand, decision tree approach stands out as highly scalable: it
takes only $4.2$ seconds for the decision-tree approach to building a classifier for
the benchmark  $\Pat_{1010101}$ with $400$ methods and $4000$ traces, while it
takes $652.4$ seconds for the max-likelihood approach to constructing the
discriminants.  
Table~\ref{table4-1} shows that the discriminants learned using
decision tree approach are simpler than the ones learned using max-likelihood
approach requiring a fewer number of tests.

\section{Case Study: Understanding Traces with Decision Trees}
\label{sec:experimental}
The data on microbenchmarks suggest that the decision tree learning
approach is more scalable and has comparable accuracy as the
max-likelihood approach. Therefore, we consider 
three case studies to evaluate whether the decision tree
approach produces useful artifacts for debugging program vulnerabilities.

\vspace{-1em}
\paragraph{Research Question.} We consider the following question: 
\vspace{-1em}
\begin{framed}
 { Does the learned discriminant pinpoint code fragments that explain
  differences in the overall execution times? }
\end{framed}
\vspace{-1em}

We consider this question to be answered positively if we can identify an
explanation for timing differences (which can help debug to side channel or
availability vulnerabilities) through \toolname \footnote{https://github.com/cuplv/Discriminer}.

\vspace{-1em}
\paragraph{Methodology.}
We consider the discriminant analysis approach based on decision tree
learning from Section~\ref{sec:mining}. 
Table~\ref{tbl:discriminant-parameters} summarizes the particular
instantiations for the discriminant analysis that we consider here. 

\begin{table}[b]
\caption{Parameters for trace set discriminant analysis, which predicts a class
label based on attributes. Here, we wish to discriminate traces to predict the
total execution time of the trace based on the methods called in the trace
and the number of times each method is called. To consider a finite number of
class labels, we fix a priori $n$ possible time ranges based on choosing the
best number of clustering.   
}
\label{tbl:discriminant-parameters}
\begin{tabular*}{\linewidth}{@{\extracolsep{\fill}}>{\bfseries}lp{0.8\linewidth}} \toprule
attributes &
  (1) the methods called in the trace (Boolean)
\\
& (2) the number of times each method is called in a trace (integer)
\\[0.5ex]
class label &
  a time range for the total execution time of the trace
\\[0.5ex]
number of classes &
  6, 6, and 2 for SnapBuddy, GabFeed, and TextCrunchr
\\ \bottomrule
\end{tabular*}
\end{table}

\smallskip\textit{Attributes: Called Methods.}
For this case study, we are interested in seeing whether the key
methods that explain the differences in execution time can be pinpointed. Thus, we
consider attributes corresponding to the called methods in a trace. 
In order to collect information regarding the called methods,
we instrumented Java bytecode applications using Javassist analysis framework
(\url{http://jboss-javassist.github.io/javassist/}).  

\smallskip\textit{Class Label: Total Execution Time Ranges.}
To identify the most salient attributes, we fix a small number of possible
labels, and cluster traces according to total execution time. 
Each cluster is defined by a corresponding time interval. The clusters and
their intervals are learned using $k$-means clustering
algorithm. 

We consider the execution time for each trace to be a random variable and assume
a normal distribution.
We obtain the mean and variance through $10$ repeated measurements. We
apply clustering to the mean execution times of each trace to
determine the class labels. Henceforth, when we speak of the
execution time of a trace, we refer to the mean of the measurements for that trace.

A class label (or cluster) can be identified by the mean of
all execution times belonging to that cluster. Then, considering the
class labels sorted in increasing order, we define the lower boundary
of a bucket for classifying new traces by averaging the maximum
execution time in the previous bucket and the minimum execution time
in this bucket (and analogously for the upper boundary).

\smallskip\textit{Weighted Labeling of Traces.}
Given a set of time ranges (clusters), we define a weighted labeling of traces
that permits a trace to be assigned to different clusters with different weights.
For a given trace, the weights to clusters are determined by the
probability mass that belongs to the time range of the cluster.
For example, consider a sample trace whose execution-time distribution
straddles the boundary of two clusters $C_0$ and $C_1$, with $22\%$ area of the
distribution intersecting with cluster $C_0$ and $78\%$ with cluster  $C_1$.
In this case, we assign the trace to both clusters $C_0$ and $C_1$  with
weights according to their probability mass in their respective
regions.
Note that this provides a smoother interpretation of the class labels
rather than assigning the most likely label.

\smallskip\textit{Decision Tree Learning.}
From a training set with this weighted labeling, we apply the weighted
decision tree learning algorithm CART described in Sec.~\ref{sec:mining}.
We use \toolname both for clustering in the time domain as described
above to determine the class labels and weights of each trace and for
learning the classification model. We use group k-fold cross
validation procedure to find accuracy.

\begin{wraptable}{R}{0pt}
\begin{tabular}{@{\extracolsep{0.5em}}lrrr} \toprule
 & total & total & observed \\
program & methods & traces & methods \\
        & (num)   & (num)  & (num)
\\ \midrule
SnapBuddy & 3071 & 439 & 160
\\
GabFeed & 573 & 368 & 30
\\
TextCrunchr & 327 & 180 & 35
\\ \midrule
\bf total & 3971 & 987 & 225
\\ \bottomrule
\end{tabular}
\end{wraptable}
\paragraph{Objects of Study.}
We consider three programs drawn from benchmarks provided by
the DARPA STAC project.
These medium-sized Java programs were developed to be 
realistic applications that may potentially have timing side-channel or
availability security vulnerabilities. 
SnapBuddy is a web application for social image sharing. The profile
page of a user includes their picture (with a filter). 
The profile page is publicly accessible. 
GabFeed is a web application for hosting community forums.
Users and servers can mutually
authenticate using public-key infrastructure. 
TextCrunchr is a text analysis program capable of performing standard text
analysis including word frequency, word length, and so on.  
It uses sorting algorithms to perform the analysis.

In the inset table, we show the basic characteristics of these benchmarks. The
benchmarks, in total, consist of 3,971 methods. From these programs, we
generated 987 traces by using a component of each applications web API (scripted
via \texttt{curl}). 
In these recorded traces, we observed 225 distinct methods called. Note that
some methods are called thousands to millions of times. 

\begin{figure}[tbp]
\begin{tabular*}{\linewidth}{@{\extracolsep{\fill}}cc}
\\
\subfloat[Time (s) of each trace]{\label{fig:snapbuddy-scatter}
\includegraphics[width=0.49\linewidth]{SnapBuddy_scatter_2.png}
}
&
\subfloat[Decision tree accuracy: 99.5$\%$]{\label{fig:snapbuddy-tree}
\includegraphics[width=0.49\linewidth]{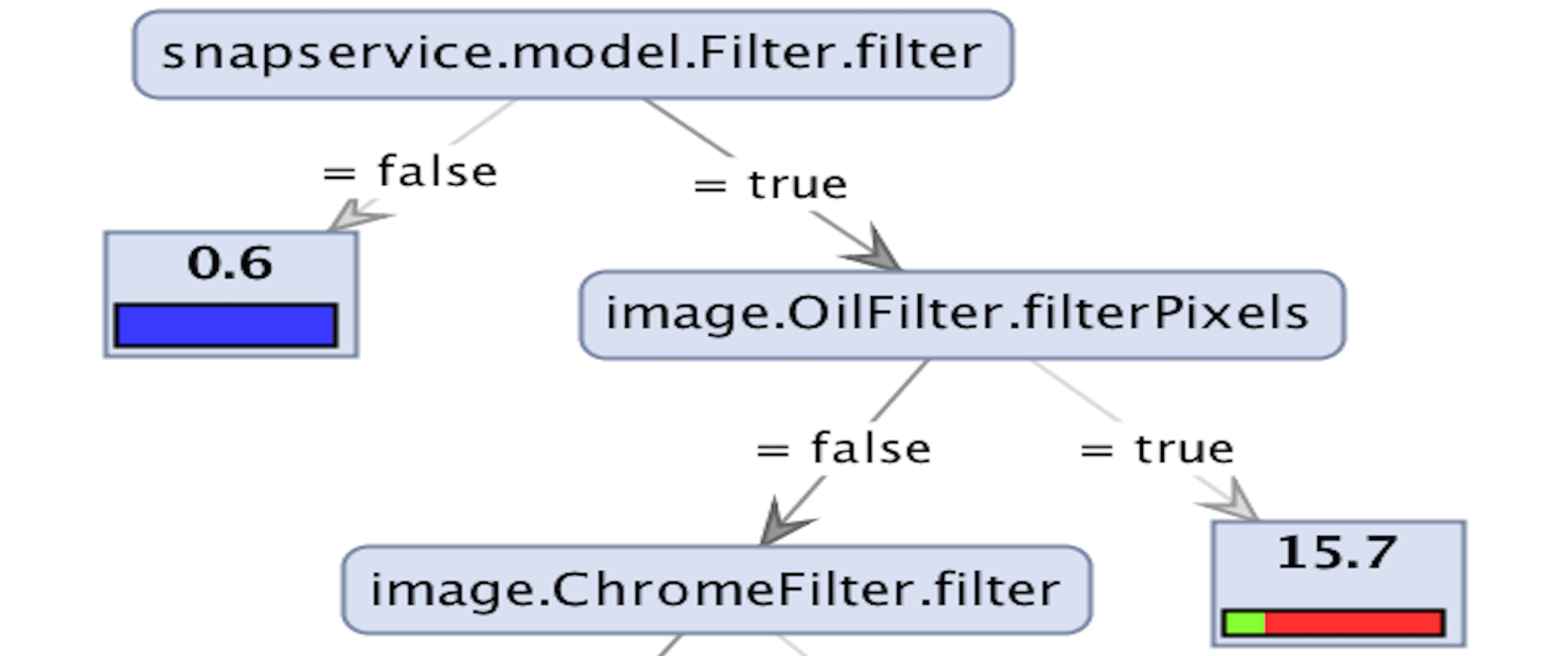}
}
\\
\subfloat[Time (s) of each trace]{\label{fig:gabfeed-scatter}
\includegraphics[width=0.49\linewidth]{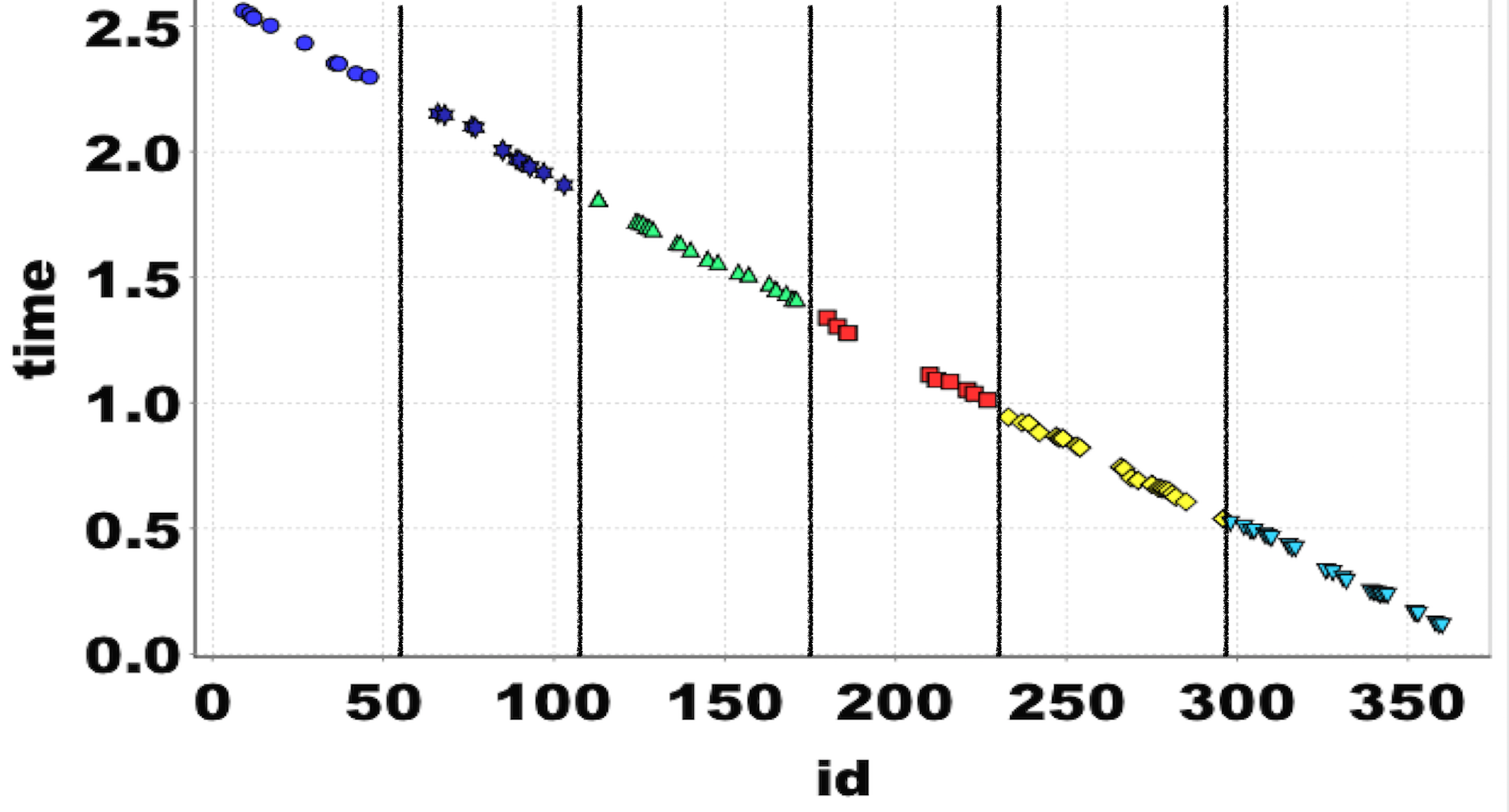}
}
&
\subfloat[Decision tree accuracy: 97.6$\%$ ]{\label{fig:gabfeed-tree}
\includegraphics[width=0.49\linewidth]{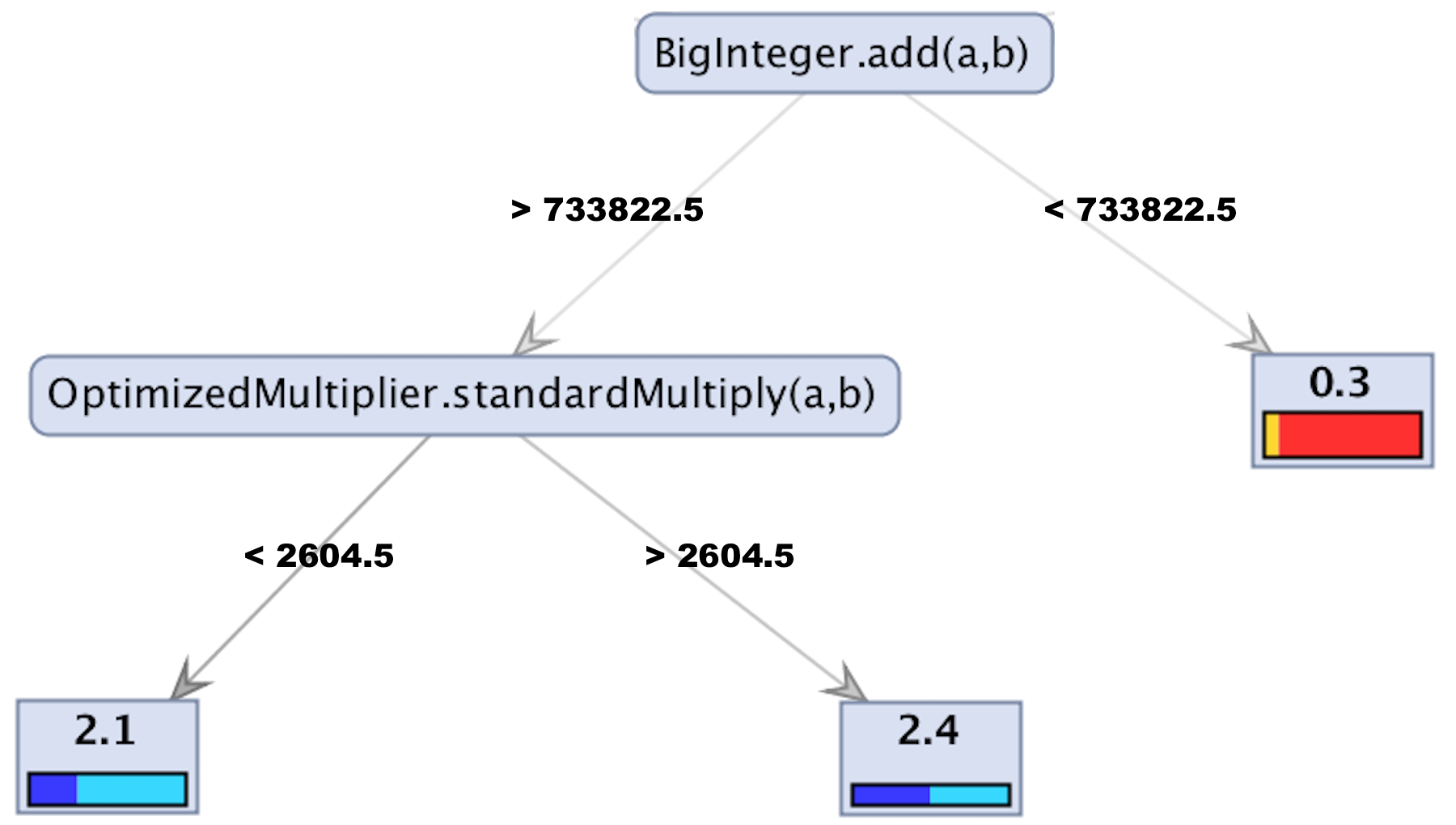}
}
\\
\subfloat[Time (s) of each trace]{\label{fig:textcrunchr-scatter}
\includegraphics[width=0.49\linewidth]{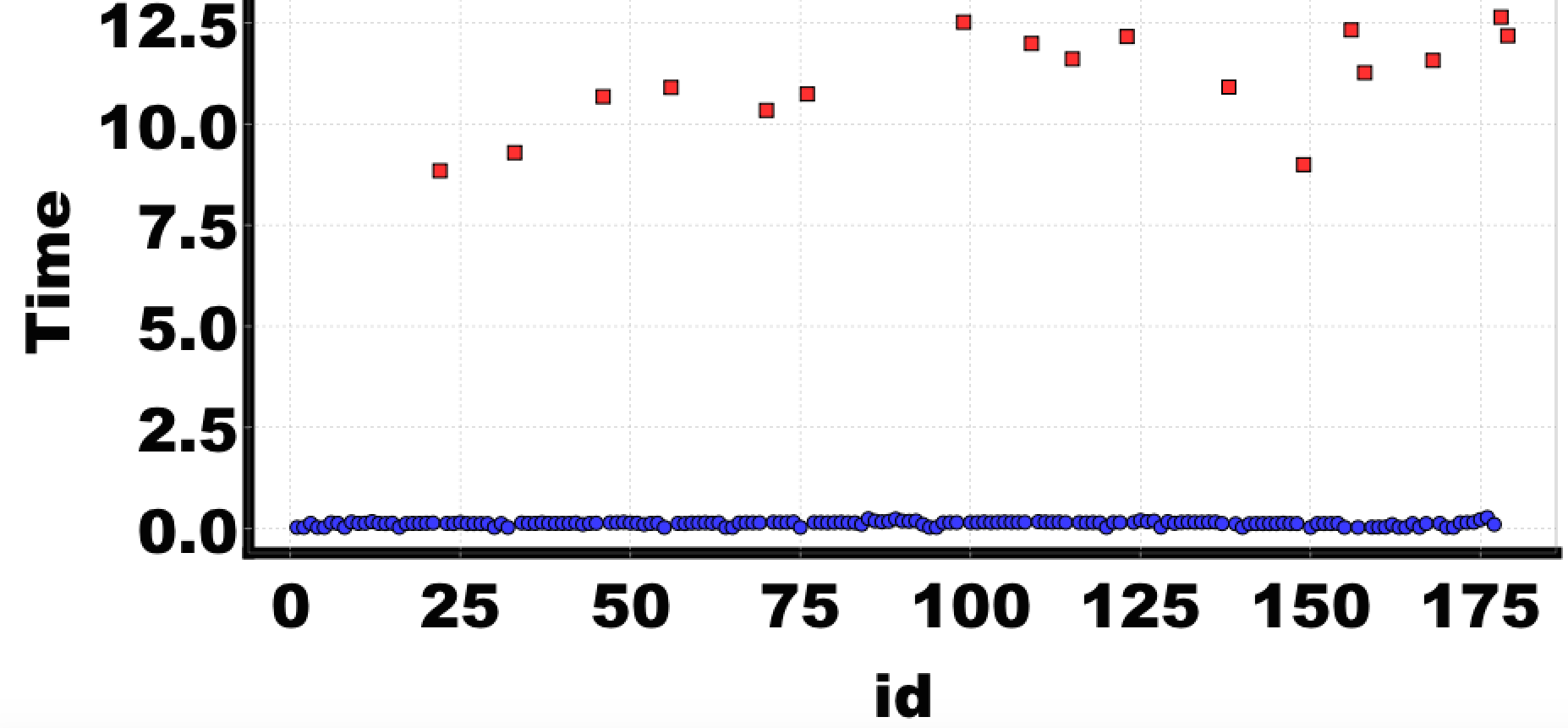}
}
&
\subfloat[Decision tree accuracy: 99.1$\%$ ]{\label{fig:textcrunchr-tree}
\includegraphics[width=0.45\linewidth]{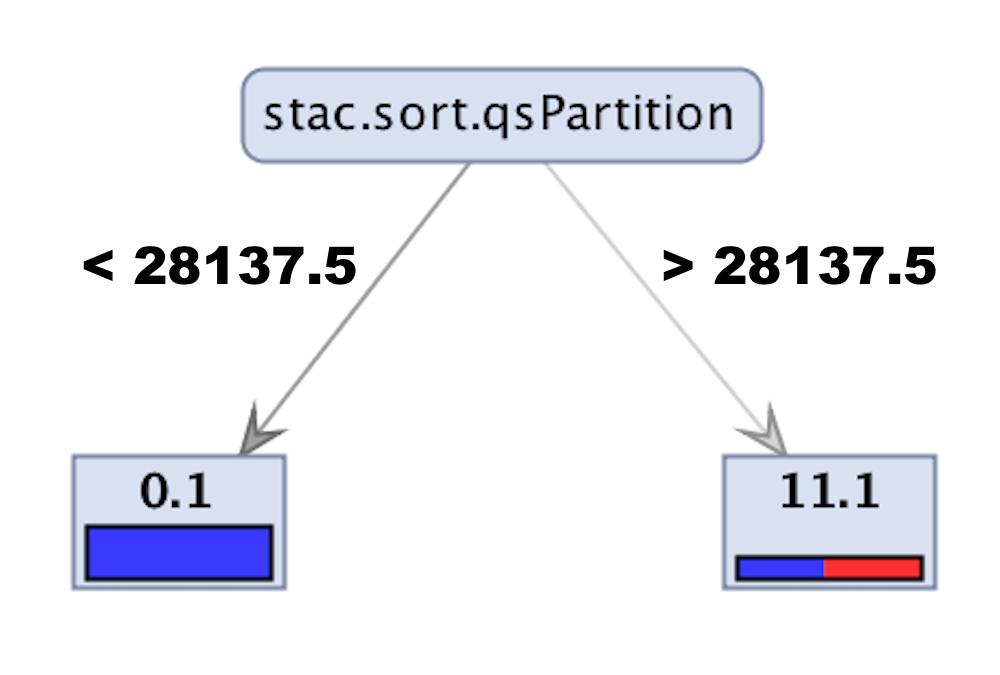}
}
\\
\end{tabular*}
\caption{Clustering in the time domain (a)-(c)-(e) to learn decision tree
classification models (b)-(d)-(f). The upper row corresponds to SnapBuddy
traces, the middle row corresponds GabFeed traces, while the bottom row corresponds to
TextCrunchr traces.   
}
\vspace{-0.3em}
\label{fig:bigfig}
\end{figure}

\paragraph{Decision Trees Produced by \toolname.}
In Fig.~\ref{fig:bigfig}\subref{fig:snapbuddy-tree}--\subref{fig:gabfeed-tree}--\subref{fig:textcrunchr-tree},
we show the decision tree learned from the SnapBuddy, GabFeed, and TextCrunchr traces, respectively.
As a decision tree is interpreted by following a path from the root to a leaf where the leaf yields the class label and the conjunction of the internal nodes describes the discriminator, one can look for characteristics of discriminated trace sets by following different paths in the tree.
The class labels at leaves are annotated with the bucket's mean
time. For example, in \subref{fig:snapbuddy-tree}, the label 15.7
shows that the path to this label which calls
\texttt{image.OilFilter.filterPixels} takes 15.7 seconds to execute.  
The colors in bars in the leaves represent the actual labels of the training traces
that would be classified in this bucket according to the learned
discriminator. Multiple colors in the bars mean that a discriminator, while
not perfectly accurate on the training traces, is also able to tolerate
noise. The height of the bar gives an
indication of the number of training traces following this
discriminator. The scatter plots
in \subref{fig:snapbuddy-scatter}--\subref{fig:gabfeed-scatter}--\subref{fig:textcrunchr-scatter}
show the time of each trace, with the color indicating the
corresponding cluster. 

\paragraph{Findings for SnapBuddy.}
For SnapBuddy, the traces exercise downloading the public profile pages of
all user from a mock database. 
We have explained in Sec.~\ref{sec:overview} how clustering (in
Fig.~\ref{fig:bigfig}\subref{fig:snapbuddy-scatter}) helps to identify
a timing side-channel, and how the decision tree (in
Fig.~\ref{fig:snapbuddy-tree}) helps in debugging the vulnerability.  

\paragraph{Findings for GabFeed.}
{\em Inputs.} For GabFeed, the traces exercise the authentication web API by fixing the user
public key and by sampling uniformly from the server private key space (3064-bit
length keys). 
{\it Identifying a Timing Side-Channel with Clustering.}
Considering scatter plot of GabFeed in Fig.~\ref{fig:gabfeed-scatter} (boundaries show different clusters),
we can see less definitive timing clusters.
However, it shows timing
differences that indicate a side channel. 
{\it Debugging Timing Side-Channels with Decision Tree Learning.}
The (part of) decision tree for GabFeed in Fig.~\ref{fig:gabfeed-tree}
is also less definitive than for SnapBuddy as we might expect given
the less well-defined execution time clusters. 
However, the part of
the decision tree
discriminants \texttt{OptimizedMultiplier.standardMultiply} for time
differences. Note that the attributes on the outgoing edge labels
correspond to a range for the number of times a particular method is
called. The decision tree explains that the different number of calls
for \texttt{OptimizedMultiplier.standardMultiply} leads to different
time buckets.  
\newsavebox{\SBoxModPow}
\begin{lrbox}{\SBoxModPow}\footnotesize
\begin{lstlisting}[language=Java,mathescape]
BigInt modPow(BigInt base, BigInt exp, BigInteger mod) { $\ldots$
  for (; i < width; i++) { $\ldots$
    if (exp.testBit(width - i - 1)) {
      s = OptimizedMultiplier.fastMultiply(s, base).mod(mod);
} $\ldots$ } $\ldots$ }
\end{lstlisting}
\end{lrbox}
By going back to the source code, we observed that
\texttt{standardMultiply} is called for each 1-bit in the server's
private key. The method {\tt standardMultiply} is called from a
modular exponentiation method called during authentication. 
What leaks is thus the number of 1s in the private key. 
A potential fix could be to rewrite the modular exponentiation method to
pad the timing differences.

\paragraph{Findings for TextCrunchr.}
{\it Inputs.}
For TextCrunchr, we provided four types of text inputs to 
analyze timing behaviors: sorted, reverse-sorted, randomly generated,
and reversed-shuffled arrays of characters (reverse-shuffle is an
operation that undoes a shuffle that TextCrunchr performs
internally). It is the reverse shuffled inputs that lead to high
execution time. 
Although the input provided to \toolname for analyzing TextCrunchr include
carefully crafted inputs (reversed shuffled sorted array), it can be argued that
a system administrator interested in auditing a security of a server has access
to a log of previous inputs including some that resulted in high execution time.
{\it Identifying Availability Vulnerabilities with Clustering.}
Considering scatter plot of TextCrunchr in
Fig.~\ref{fig:textcrunchr-scatter} 
we can see well-defined timing clusters which can potentially lead to
security issues. It shows that a small fraction of inputs
takes comparably higher time of execution in comparison to the
others. Thus an attacker can execute a denial-of-service
(availability) attack by repeatedly providing the costly inputs (for
some inputs, it will take more than 600 seconds to process the text). 
The system administrator mentioned above probably knew from his logs about possible
inputs with high execution time. What he did not know is why these
inputs lead to high execution time. 
{\it Debugging Availability Vulnerabilities with Decision Tree Learning.}
The decision tree for TextCrunchr in Fig.~\ref{fig:textcrunchr-tree}
shows that 
the number of calls on \texttt{stac.sort.qsPartition} as the explanation for
time differences (out of 327 existing methods in the application). 
This can help identify the sorting algorithm (Quicksort) used as a source of the
problem and leads to the realization that certain inputs trigger the
worst-case execution time of Quicksort. 

\paragraph{Threats to Validity.}
These case studies provide evidence that decision tree learning helps 
in identifying code fragments that correlate with differential execution time. 
Clearly, the most significant threat to validity is whether these 
programs are representative of other applications. 
To mitigate, we considered programs not created by
us nor known to us prior to this study. 
These applications were 
designed to faithfully represent real-world Java programs---for
example, using Java software engineering patterns and best
practices. 
Another threat concerns the representativeness of the training
sets. To mitigate this threat, we created sample traces directly using the web
interface for the whole application, rather than interposing at any intermediate
layer. 
This interface is for any user of these web applications and specifically the
interface available to a potential attacker. 
A training set focuses on exercising a particular feature of the application,
which also corresponds to the ability of an attacker to build training sets
specific to different features of the application.

\section{Related Work}
\label{sec:related}
Machine learning techniques have been used for {\em specification mining}, that
is, for learning succinct representations of the set of all program traces. 
Furthermore, machine learning techniques have been applied to learn classifiers of
programs for {\em malware detection} and for 
{\em software bug detection}. 

\paragraph{Specification Mining.}
In~\cite{AMM02}, machine learning
techniques are used to synthesize an NFA (nondeterministic finite
automaton) that represents all the correct traces of a
program. In our setting, this would correspond to learning a
discriminant for one cluster (of correct traces). In contrast, our decision trees
discriminate multiple clusters. However, the discriminants we
considered in this paper are less expressive than NFAs. The survey~\cite{Zeller11}
provides an overview of other specification mining approaches. 

\paragraph{Malware and Bug Detection.}
In malware detection, machine learning techniques are used to learn
classifiers that classify programs into benign and
malicious~\cite{RIE08,BAI07,BUR11,AAF13,WU12,KOL09,FRE10}.
In software bug detection, the task is to learn
classifiers that classify programs behaviors into faulty and
non-faulty~\cite{SUN10,LO09,WEI05,ELI08}. In contrast, we consider
more clusters of traces. 
In particular, \citet{LO09} constructs a classifier to generalize known
failures of software systems and to further detect (predict) other
unknown failures. First, it mines iterative patterns from program
traces of known normal and failing executions. Second, it applies a
feature selection method to identify highly
discriminative patterns which distinguish failing traces from
normal ones.

\noindent In all these works, the training set is labeled: all the programs are
labeled either benign or malicious (faulty or non-faulty). In
contrast, we start with an unlabeled set of traces, and 
construct their labels by clustering in the time domain.

\section{Conclusion}
\label{sec:concl}

\paragraph{Summary.}
We introduced the trace set discrimination problem as a formalization
of the practical problem of finding what can be inferred from limited run
time observations of the system. We have shown that the problem is
{\sc NP}-hard, and have proposed two scalable techniques to solve it. The
first is ILP-based, and it can give formal guarantees about the
discriminant that was found but infers discriminants of a limited
form. The second is based on decision trees, infers general
discriminants, but does not give formal guarantees. 
For three realistic applications, our tool produces a
decision tree useful for explaining timing differences between executions.

\paragraph{Future Work.} There are several intriguing directions for
future research. First, we will
investigate the extension of our framework to reactive systems, by
generalizing our notion of execution time observations to sequences of timed
events. Second, we will build up
the network traffic monitoring ability of our tool, to make it usable
by security analysts for distributed architectures.

\bibliographystyle{plainnat}
\bibliography{papers}
\end{document}